%% file: main.tex
\documentclass[letterpaper, 10pt]{article}
\usepackage[margin=1in]{geometry}
\usepackage[noend]{algpseudocode}
\usepackage{boxedminipage}
\usepackage{enumitem}
\usepackage{amsthm}
\usepackage{amsmath,amssymb}
\usepackage{xcolor}
\usepackage{url}
\usepackage{xspace}
\usepackage{graphicx}
\usepackage{pgf}
\usepackage{subcaption}
\usepackage{multirow}
\usepackage{authblk}

\newtheorem{theorem}{Theorem}
\newtheorem{lemma}[theorem]{Lemma}

\theoremstyle{definition}

\newcommand{\name}{Flexible BFT\xspace}
\renewcommand{\paragraph}[1]{\vspace{7pt}\noindent\textbf{#1}}
\input{setup}

\ignore{
\fancyhf{} 
\fancyhead[C]{Anonymous submission \#9999 to ACM CCS 2019} 
\fancyfoot[C]{\thepage}

\setcopyright{none} 
\acmConference[Anonymous Submission to ACM CCS 2019]{ACM Conference on Computer and Communications Security}{Due 15 May 2019}{London, TBD}
\acmYear{2019}

\settopmatter{printacmref=false, printccs=true, printfolios=true}
}

\begin{document}
\title{Flexible Byzantine Fault Tolerance}
\author[1]{Dahlia Malkhi}
\author[1,2]{Kartik Nayak}
\author[1,3]{Ling Ren}
\affil[1]{VMware Research -- {\tt \{dmalkhi,nkartik,lingren\}@vmware.com}}
\affil[2]{Duke University}
\affil[3]{University of Illinois at Urbana-Champaign}
\ignore{
}
\sloppy
\maketitle
\begin{abstract}
This paper introduces \name, a new approach for BFT consensus solution design revolving
around two pillars, stronger resilience and diversity.
The first pillar, stronger resilience, involves a new fault model called \pleasing faults.
\Pleasing replicas may arbitrarily deviate from the protocol in an attempt to break safety of the protocol. 
However, if they cannot break safety, they will not try to prevent liveness of the protocol.
Combining \pleasing faults into the model, 
\name is resilient to higher corruption levels than possible in a pure Byzantine
fault model. 
%
The second pillar, diversity, designs consensus solutions 
whose protocol transcript is used to draw different commit decisions under diverse beliefs. 
With this separation, the same \name solution supports synchronous and asynchronous
beliefs, as well as varying resilience threshold combinations of Byzantine and
\pleasing faults.

At a technical level, \name achieves the above results using two new ideas. 
First, it introduces a synchronous BFT protocol 
in which only the commit step requires to know the network delay bound 
and thus replicas execute the protocol without any synchrony assumption.
Second, it introduces a notion called Flexible Byzantine Quorums by 
dissecting the roles of different quorums in existing consensus protocols.
\end{abstract}

\ignore{

}
\input{intro}

\input{model}

\input{synchrony-overview}

\input{overview}
\input{protocol}

\input{discussion}
\input{related}

\input{conclusion}

\section*{Acknowledgement}
We thank Ittai Abraham and Ben Maurer for many useful discussions
on Flexible BFT. We thank Marcos Aguilera for many insightful
comments on an earlier draft of this work
\bibliographystyle{plain}
\bibliography{refs}
\end{document}

%% file: setup.tex
\newcommand{\kartik}[1]{}
\newcommand{\ignore}[1]{}

\newcommand{\sig}[1]{\langle #1 \rangle}

\newcommand{\Status}{\mathsf{status}}
\newcommand{\Propose}{\mathsf{propose}}

\newcommand{\Blame}{\mathsf{blame}}

\newcommand{\Vote}{\mathsf{vote}}

\newcommand{\TimeEquiv}{{\mathsf{t}\text{-}\mathsf{equiv}}}
\newcommand{\TimeAccept}{{\mathsf{t\text{-}{lock}}}}
\newcommand{\TimeBlame}{{\mathsf{t\text{-}{viewchange}}}}
\newcommand{\TimeCurrent}{{\mathsf{current}\text{-}\mathsf{time}}}

\newcommand{\TimeLock}{{\mathsf{t\text{-}{lock}}}}

\newcommand{\CommitCert}{\mathcal{C}^{\qmin}}
\newcommand{\Cert}{\CommitCert}
\newcommand{\CertN}{\mathcal{C}}

\newcommand{\StatusCert}{\mathcal{S}}
\newcommand{\pleasing}{alive-but-corrupt\xspace}
\newcommand{\psing}{a-b-c\xspace}
\newcommand{\Pleasing}{Alive-but-corrupt\xspace}

\newcommand{\block}{B}

\newcommand{\ViewNumber}{v}

\newcommand{\q}{{q_c}}
\newcommand{\qmin}{{q_{r}}}
\newcommand{\qlock}{{q_\text{lck}}}
\newcommand{\qunlock}{{q_\text{ulck}}}
\newcommand{\qunique}{{q_\text{unq}}}
\newcommand{\qcommit}{{q_\text{cmt}}}


%% file: intro.tex
\section{Introduction}
\label{sec:introduction}

Byzantine fault tolerant (BFT) protocols are used to build replicated services~\cite{PSL80,LSP82,schneider1990implementing}.
Recently, they have received revived interest
as the algorithmic foundation of what is known as decentralized ledgers, or blockchains. 

In the classic approach to BFT protocol designs,
a protocol designer or a service administrator first picks a set of assumptions 
(e.g., the fraction of Byzantine faults and certain timing assumptions)
and then devises a protocol (or chooses an existing one) tailored for that particular setting.
The assumptions made by the protocol designer are imposed upon all parties involved --- 
every replica maintaining the service as well as every client (also known as the
"learner" role) using the service.
Such a protocol collapses if deployed under settings that differ from the one it is designed for. 
In particular, optimal-resilience partially synchronous solutions~\cite{DLS88,castro1999practical} 
break (lose safety and liveness) if the fraction of Byzantine faults exceeds
$1/3$. 
Similarly, optimal-resilience synchronous solutions~\cite{abraham2018synchronous,hanke2018dfinity}
do not obtain safety or liveness if the fraction of Byzantine faults exceeds
$1/2$ or if the synchrony bound is violated.

In this work, we introduce a new approach for BFT protocol design called \emph{\name}. 
Our approach offers advantages in the two aspects above.
First, the \name approach enables protocols 
that tolerate more than $1/3$ (resp.\ $1/2$) corruption faults in the
partial-synchrony (resp.\ synchrony) model
--- provided that the number of Byzantine faults do not exceed the respective resilience bounds. 
Second, the \name approach allows a certain degree of separation between 
the fault model and the protocol design.
As a result, \name allows diverse clients with different 
fault assumptions and timing assumptions (synchrony or not)
to participate in the same protocol.
We elaborate on these two aspects below.

\paragraph{Stronger resilience.}
We introduce a mixed fault model with a new type of fault called
\emph{\pleasing} (\psing for short) faults. 
\Pleasing replicas actively try to disrupt the system from maintaining a safe consensus decision 
and they might arbitrarily deviate from the protocol for this purpose. 
However, if they cannot break safety,
they will not try to prevent the system from reaching a (safe) decision.
The rationale for this new type of fault is that 
violating safety may provide the attacker gains (e.g., a double spend attack)
but preventing liveness usually does not.
In fact, \psing replicas may gain rewards from keeping the replicated service
live, e.g., by collecting service fees.
We show a family of protocols 
that tolerate a combination of Byzantine and \psing faults 
that exceeds $1/3$ in the partially synchronous model
and exceeds $1/2$ in the synchronous model.
Our results do not violate existing resilience bounds 
because the fraction of Byzantine faults is always smaller than the respective bounds.

\paragraph{Diversity.}
The \name approach further provides certain separation between the fault model
and the protocol. The design approach builds a protocol whose transcript can be
interpreted by external clients with diverse beliefs, who draw different consensus
commit decisions based on their beliefs.
\name guarantees safety and liveness so far as the
clients' beliefs are correct; thus two clients with correct
assumptions agree with each other.
Clients specify (i) the fault threshold they need to tolerate, 
and (ii) the message delay bound, if any, they believe in.  
For example, one instance of \name can support a client that requires tolerance against $1/5$ Byzantine faults plus $3/10$ \psing faults, 
while simultaneously supporting another client who requires tolerance against $1/10$ Byzantine faults plus $1/2$ \psing faults,
and a third client who believes in synchrony and requires 
$3/10$ Byzantine plus $2/5$ \psing tolerance.

This novel separation of fault model from protocol design
can be useful in practice in several ways.
First, different clients may naturally hold different assumptions about the system.
Some clients may be more cautious and require a higher resilience than others;
some clients may believe in synchrony while others do not.
Moreover, even the same client may assume a larger fraction of faults 
when dealing with a \$1M transaction compared to a \$5 one. 
The rationale is that more replicas may be willing to collude
to double spend a high-value transaction.
In this case, the client can wait for more votes before committing the \$1M transaction.
Last but not least, a client may update its assumptions based on certain events it observes.
For example, if a client receives votes for conflicting values,
which may indicate an attempt at attacking safety,
it can start requiring more votes than usual;
if a client who believes in synchrony notices abnormally long message delays,
which may indicate an attack on network infrastructure,
it can update its synchrony bound to be more conservative or switch to a partial-synchrony assumption.

The notion of ``commit'' needs to be clarified in our new model.
Clients in \name have different assumptions and hence different commit rules.
It is then possible and common that a value is committed by one client but not another.
\name guarantees that any two clients whose assumptions are correct 
(but possibly different) commit to the same value.
If a client's assumption is incorrect, however,
it may commit inconsistent values which may later be reverted.
While this new notion of commit may sound radical at first, 
it is the implicit behavior of existing BFT protocols.
If the assumption made by the service administrator is violated in a classic BFT protocol
(e.g., there are more Byzantine faults than provisioned),
clients may commit to different values and they have no recourse.
In this sense, \name is a robust generalization of classic BFT protocols. 
In \name, if a client performs conflicting commits,
it should update its assumption to be more cautious
and re-interpret what values are committed under its new assumption.
In fact, this ``recovery'' behavior is somewhat akin to Bitcoin.
A client in Bitcoin decides how many confirmations are needed 
(i.e., how ``deeply buried'') to commit a block.
If the client commits but subsequently an alternative longer fork appears, its commit is reverted.
Going forward, the client may increase the number of confirmations it requires.

\paragraph{Key techniques.}
\name centers around two new techniques.
The first one is a novel synchronous BFT protocol with replicas
 executing at \emph{network speed}; 
that is, the protocol run by the replicas does not assume synchrony.
This allows clients in the same protocol 
to assume different message delay bounds and commit at their own
pace. 
The protocol thus 
 separates timing assumptions of replicas from timing assumptions of clients.
Note that this is only possible via \name's separation of
protocol from the fault model:
the action of committing is only carried out by clients, not by replicas.
The other technique involves a breakdown of the different roles that 
quorums play in different steps of partially synchronous BFT protocols.
Once again, made possible by the separation in \name,
we will use one quorum size for replicas to run a protocol,
and let clients choose their own quorum sizes for committing in the protocol.

\paragraph{Contributions.}
To summarize, our work has the following contributions.

\begin{enumerate}[topsep=8pt,itemsep=8pt]
				
\item 
\textbf{Alive-but-corrupt faults.} We introduce a new type of fault, called
\pleasing fault, which attack safety but not liveness.

\item \textbf{Synchronous BFT with network speed replicas.} 
	We present a synchronous protocol in which 
	only the commit step requires synchrony.
	Since replicas no longer perform commits in our approach,
	the protocol simultaneously supports clients assuming different synchrony bounds.
		
\item \textbf{Flexible Byzantine Quorums.} 
	We deconstruct existing BFT protocols to understand the role played by different quorums 
	and introduce the notion of Flexible Byzantine Quorums.
	A protocol based on Flexible Byzantine Quorums simultaneously supports clients assuming different fault models.

\item \textbf{One BFT Consensus Solution for the Populace.} 
	Putting the above together, we present a new approach for BFT design, \name.
	Our approach has stronger resilience and diversity: \name tolerates a fraction of combined (Byzantine plus \psing) faults 
	beyond existing resilience bounds. And clients with diverse fault
and timing beliefs are supported  in the same protocol. 

\end{enumerate}

\paragraph{Organization.}
The rest of the paper is organized as follows.
Section~\ref{sec:model} defines the \name model where replicas and clients are separated.
We will describe in more detail our key techniques 
for synchrony and partial-synchrony in
Sections~\ref{sec:overv-synchr-flex}~and~\ref{sec:overview-async}, respectively. 
Section~\ref{sec:protocol} puts these techniques together and presents the final protocol.
Section~\ref{sec:discussion} discusses the result obtained by the
Flexible BFT design and 
Section~\ref{sec:related-work} describes related work.


%% file: model.tex
\section{Modeling \name}
\label{sec:model}

The goal of \name is to build a replicated service that takes requests from clients
and provides clients an interface of a single non-faulty server,
i.e., it provides clients with the same totally ordered sequence of values. 
Internally, the replicated service uses multiple servers,
also called replicas, to tolerate some number of faulty servers. 
The total number of replicas is denoted by $n$. In this paper, whenever we
speak about a set of replicas or messages, we denote the set size as its fraction over $n$.
For example, we refer to a set of $m$ replicas as ``$q$ replicas'' where $q=m/n$.

Borrowing notation from Lamport~\cite{lamport2006fast}, 
such a replicated service has three logical actors: 
\emph{proposers} capable of sending new values,
\emph{acceptors} who add these values to a totally ordered sequence (called a blockchain), 
and \emph{learners} who decide on a sequence of values based on the transcript of the protocol and execute them on a state machine. 
Existing  replication protocols provide the following two properties:

\begin{description}[topsep=8pt,itemsep=4pt]
\item[-] \textbf{Safety.} Any two learners learn the same sequence of values.
\item[-] \textbf{Liveness.} A value proposed by a proposer will
  eventually be executed by every learner.
\end{description}

In existing replication protocols, 
the learners are assumed to be \emph{uniform}, 
i.e., they interpret a transcript using the same rules 
and hence decide on the same sequence of values.
In \name, we consider diverse learners with different assumptions. 
Based on their own assumptions, 
they may interpret the transcript of the protocol differently. 
We show that so far as the assumptions 
of two different learners are both correct, 
they will eventually learn the same sequence of values.
A replication protocol in the \name approach
satisfies the following properties:

\begin{description}[topsep=8pt,itemsep=4pt]
\item[-] \textbf{Safety for diverse learners.} Any two
  learners with correct but potentially different assumptions
  learn the same sequence of values.
\item[-] \textbf{Liveness for diverse learners.} A value
  proposed by a proposer will eventually be executed by every learner
  with a correct assumption.
\end{description}

In a replicated service, clients act as proposers and learners, 
whereas the replicas (replicated servers) are acceptors.
Thus, safety and liveness guarantees are defined with respect to clients.

\paragraph{Fault model.} 
We assume two types of faults within the replicas: Byzantine and \emph{\pleasing} (\psing for short). 
Byzantine replicas behave arbitrarily. 
On the other hand, the goal of \psing replicas
is to attack safety but to preserve liveness.
These replicas will take any actions that help them break safety of the protocol.
However, if they cannot succeed in breaking safety, they will help provide liveness.
Consequently, in this new fault model, 
the safety proof should treat \psing replicas similarly to Byzantine.
Then, \emph{once safety is proved},
the liveness proof can treat \psing replicas similarly to honest.
We assume that the adversary is static, i.e., the adversary determines which
replicas are Byzantine and \psing before the start of the protocol. 

\paragraph{Other assumptions.} 
We assume hash functions, digital signatures and a public-key infrastructure (PKI).
We use $\sig{x}_R$ to denote a message $x$ signed by a replica $R$. 
We assume pair-wise communication channels between replicas. 
We assume that all replicas have clocks that advance at the same
rate.


%% file: synchrony-overview.tex
\section{Synchronous  BFT with Network Speed Replicas - Overview}
\label{sec:overv-synchr-flex}
\begin{figure*}[ht]
\begin{boxedminipage}{\textwidth}
\paragraph{Protocol executed by the replicas.} 
\begin{enumerate}
\setlength\itemsep{0.5em}

\item \textbf{Propose. } \label{step:sync-propose} 
  The leader $L$ of view $v$ 
  proposes a value $b$.
\item \textbf{Vote.} \label{step:sync-vote}
  On receiving the first value $b$ in a view $v$, a replica broadcasts $b$ and votes for $b$ if it is \emph{safe} to do so, as determined by a locking mechanism described later. 
The replica records the following.
\begin{itemize}
\item[-] If the replica collects $\qmin$ votes on $b$, 
		denoted as $\CertN^\qmin_v(b)$ and called a certificate of $b$ from view $v$, 
		then it ``locks'' on $b$ and records the lock time as $\TimeLock_v$.
\item[-] If the replica observes an equivocating value signed by $L$ at any time after entering view $v$, it records the time of equivocation as $\TimeEquiv_v$. It blames the leader by broadcasting $\sig{\Blame, v}$ and the equivocating values.
\item[-] If the replica does not receive a proposal for sufficient time in view $v$, it times out and broadcasts $\sig{\Blame, v}$. 
\item[-] If the replica collects a set of $\qmin$ $\sig{\Blame, v}$ messages, 
it records the time as $\TimeBlame_v$, broadcasts them and enters view $v+1$. 
\end{itemize}
\end{enumerate}

If a replica locks on a value $b$ in a view, then it votes only for $b$ in subsequent views 
unless it ``unlocks'' from $b$ by learning that $\qmin$ replicas
are not locked on $b$ in that view or higher views (they may be locked on other values or they may not be locked at all).

\paragraph{Commit rules for clients.}
A value $b$ is said to be committed by a client assuming $\Delta$-synchrony 
iff $\qmin$ replicas
each report that there exists a view $v$ such that, 
\begin{enumerate}
\item $b$ is certified, i.e., $\CertN^\qmin_v(b)$ exists.
\item the replica observed an undisturbed-$2\Delta$ period after certification, i.e., 
		no equivocating value or view change was observed at a time before $2\Delta$ after it was certified, 
		or more formally, $\min(\TimeCurrent, \TimeEquiv_v, \TimeBlame_v) - \TimeLock_v \geq 2\Delta$
\end{enumerate}
\end{boxedminipage}
\caption{Synchronous BFT with network speed replicas.}
\label{fig:sync-bft}
\end{figure*}

Early synchronous protocols~\cite{DS83,katz2009expected,micali2017optimal} have relied on synchrony in two ways.
First, the replicas assume a maximum network delay $\Delta$ for communication between them. Second, they require a lock step execution, i.e., all replicas are in the same round at the same time.
Hanke et al. showed a synchronous protocol without lock step execution~\cite{hanke2018dfinity}. 
Their protocol still contains a synchronous step in which all replicas perform a blocking wait of $2\Delta$ time before proceeding to subsequent steps. 
Sync HotStuff~\cite{abraham2019sync} improves on it further to remove replicas' blocking waits during good periods (when the leader is honest),
but blocking waits are still required by replicas during bad situations (view changes).

In this section, we show a synchronous protocol where the
replicas do not ever have blocking waits and execute at the network speed.
In other words, replicas run a partially synchronous protocol and do not rely on synchrony at any point. 
Clients, on the other hand, rely on synchrony bounds to commit.
This separation is what allows our protocol to support clients with different assumptions on the value of $\Delta$. 
To the best of our knowledge, this is the first synchronous protocol to achieve such a separation.
In addition, the protocol tolerates a combined Byzantine plus \psing fault ratio greater than a half (Byzantine fault tolerance is still less than half). 

For simplicity, in this overview, we show a protocol for single shot consensus.
In our final protocol in Section~\ref{sec:protocol}, we will consider a pipelined version of the protocol for consensus on a sequence of values. 
We do not consider termination for the single-shot consensus protocol in this overview 
because our final replication protocol is supposed to run forever. 

The protocol is shown in Figure~\ref{fig:sync-bft}. It runs in a sequence of views. Each view has a designated leader who may be selected in a round robin order. The leader drives consensus in that view.
In each view, the protocol runs in two steps -- propose and vote. In the propose step, the leader proposes a value $b$. In the vote step, replicas vote for the value if it is \emph{safe} to do so.
The vote also acts as a \emph{re-proposal} of the value.
If a replica observes a set of $\qmin$ votes on $b$, called a
certificate $\Cert(b)$, it ``locks'' on $b$. 
For now, we assume $\qmin = 1/2$.
(To be precise, $\qmin$ is slight larger than 1/2, e.g., $f+1$ out of $2f+1$.)
We will revisit the choice
of $\qmin$ in Section~\ref{sec:discussion}.
In subsequent views, a replica will not vote for a value other than $b$ unless it learns that $\qmin$ replicas are not locked on $b$. In addition, the replicas switch views (i.e., change leader) if they either observe an equivocation or if they do not receive a proposal from the leader within some timeout. 
A client commits $b$ if $\qmin$ replicas state that there exists a view in which $b$ is certified 
and no equivocating value or view change was observed at a time before $2\Delta$ after it was certified. 
Here, $\Delta$ is the maximum network delay the client believes in. 

The protocol ensures safety if there are fewer than $\qmin$ faulty replicas.
The key argument for safety is the following: 
If an honest replica $h$ satisfies the commit condition for some value $b$ in a view, then 
(a) no other value can be certified and 
(b) all honest replicas are locked on $b$ at the end of that view.
To elaborate, satisfying the commit condition implies that some honest replica $h$ has observed an undisturbed-$2\Delta$ period after it locked on $b$, i.e., it did not observe an equivocation or a view change. 
Suppose the condition is satisfied at time $t$. 
This implies that other replicas did not observe an equivocation or a view change before $t-\Delta$. 
The two properties above hold if the quorum honesty conditions described below hold.
For liveness, if Byzantine leaders equivocate or do not propose a safe value, they will be blamed by both honest and \psing replicas and a view change will ensue. 
Eventually there will be an honest or \psing leader to drive consensus if quorum availability holds.

\begin{description}[topsep=8pt,itemsep=4pt]
\item[Quorum honesty (a) within a view.] 
Since the undisturbed period starts after $b$ is certified, 
$h$ must have voted (and re-proposed) $b$ at a time earlier than $t-2\Delta$. 
Every honest replica must have received $b$ before $t - \Delta$.
Since they had not voted for an equivocating value by then, they must have voted for $b$.
Since the number of faults is less than $\qmin$,
every certificate needs to contain an honest replica's vote.
Thus, no certificate for any other value can be formed in this view. 
		
\item[Quorum honesty (b) across views.] 
$h$ sends $\CertN^\qmin_v(b)$ at time $t-2\Delta$. 
All honest receive $\CertN^\qmin_v(b)$ by time $t-\Delta$ and become locked on $b$. 
For an honest replica to unlock from $b$ in subsequent views, 
$\qmin$ replicas need to claim that they are not locked on $b$. 
At least one of them is honest 
and would need to falsely claim it is not locked, which cannot happen.

\item[Quorum availability.]	Byzantine replicas do not exceed $1-\qmin$ so that $\qmin$ replicas respond to the leader. 
\end{description}

\paragraph{Tolerating \psing faults.}
If we have only honest and Byzantine replicas (and no \psing replicas), quorum honesty requires the fraction of Byzantine replicas $B < \qmin$. Quorum availability requires $B \leq 1-\qmin$. If we optimize for maximizing $B$, we obtain $\qmin \geq 1/2$.
Now, suppose $P$ represents the fraction of \psing replicas. Quorum honesty requires $B + P < \qmin$, and quorum availability requires $B \leq 1-\qmin$. Thus, the protocol supports varying values of $B$ and $P$ at different values of $\qmin > 1/2$ such that safety and liveness are both preserved.

\paragraph{Separating client synchrony assumption from the replica protocol.}
The most interesting aspect of this protocol is the separation of the client commit rule from
the protocol design. In particular, although this is a synchronous protocol, the replica protocol does not rely on any synchrony bound. This allows clients to choose their own message delay bounds. 
Any client that uses a correct message delay bound enjoys safety. 


%% file: overview.tex
\section{Flexible Byzantine Quorums for Partial Synchrony - Overview}
\label{sec:overview-async}

In this section, we explain the high-level insights of Flexible
Byzantine Quorums in \name. 
Again, for ease of exposition, we
focus on a single-shot consensus and do not consider termination.  
We start by reviewing the Byzantine Quorum
Systems~\cite{Malkhi:1997:BQS:258533.258650} that underlie
existing partially synchronous protocols  
that tolerate 1/3 Byzantine faults (Section~\ref{sec:byz-quorums}). 
We will illustrate that multiple uses of 2/3-quorums 
actually serve different purposes in these protocols.
We then generalize these protocols to use  
\emph{Flexible Byzantine Quorums}~(Section~\ref{sec:flexible-byz-quorum}),
the key idea that enables more than 1/3 fault tolerance 
and allows diverse clients with varying assumptions to co-exist. 

\subsection{Background: Quorums in PBFT}
\label{sec:byz-quorums}

Existing protocols for solving consensus in the partially synchronous setting 
with optimal $1/3$-resilience
revolve around voting by \emph{Byzantine quorums} of replicas. 
Two properties of Byzantine quorums are utilized for achieving safety and liveness.
First, any two quorums intersect at one honest replica -- quorum intersection.
Second, there exists a quorum that contains no Byzantine faulty replicas -- quorum availability.
Concretely, when less than $1/3$ the replicas are Byzantine,
quorums are set to size $\qmin=2/3$.
(To be precise, $\qmin$ is slightly larger than 2/3, 
i.e., $2f+1$ out of $3f+1$ where $f$ is the number of faults,
but we will use $\qmin=2/3$ for ease of exposition.)
This guarantees an intersection of size at least $2\qmin-1=1/3$, 
hence at least one honest replica in the intersection.
As for availability, there exist $\qmin=2/3$ honest replicas to form
a quorum.

To dissect the use of quorums in BFT protocols, consider their use in
PBFT~\cite{castro1999practical} for providing safety and liveness. 
PBFT operates in a view-by-view manner.
Each view has a unique leader and consists of the following steps: 

\begin{itemize}

\item[-] \textbf{Propose.} 
  A leader $L$ proposes a value $b$.

\item[-] \textbf{Vote 1.} 
  On receiving the first value $b$ for a view $v$, a replica votes for $b$ 
  if it is \emph{safe}, as determined by a locking mechanism described below.
  A set of $\qmin$ votes form a certificate $\CertN^{\qmin}(b)$.

\item[-] \textbf{Vote 2.} 
  On collecting $\CertN^{\qmin}(b)$,
  a replica ``locks'' on $b$ and votes for $\CertN^{\qmin}(b)$.

\item[-] \textbf{Commit.} 
  On collecting $\qmin$ votes for $\CertN^{\qmin}(b)$, a client learns that proposal
$b$ becomes a committed decision. 

\end{itemize}
If a replica locks on a value $b$ in a view, then it votes only for $b$ in subsequent views 
unless it ``unlocks'' from $b$.
A replica ``unlocks'' from $b$ if it learns that $\qmin$ replicas are \emph{not}
locked on $b$ in that view or higher 
(they may be locked on other values or they may not be locked at all).

The properties of Byzantine quorums are harnessed in PBFT for safety and liveness as follows:

\begin{description}[topsep=8pt,itemsep=4pt]

\item[Quorum intersection within a view.]
Safety within a view is ensured by the first round of votes.
A replica votes only once per view.
For two distinct values to both obtain certificates, 
one honest replica needs to vote for both, which cannot happen.

\item[Quorum intersection across views.]
Safety across views is ensured by the locking mechanism.
If $b$ becomes a committed decision in a view, 
then a quorum of replicas lock on $b$ in that view.
For an honest replica among them to unlock from $b$,
a quorum of replicas need to claim they are not locked on $b$.
At least one replica in the intersection is honest 
and would need to falsely claim it is not locked, which cannot happen.

\item[Quorum availability within a view.]
Liveness within each view is guaranteed by having an honest quorum respond to a
non-faulty leader.
\end{description}

\subsection{Flexible Byzantine Quorums}
\label{sec:flexible-byz-quorum}

Our \name approach separates the quorums used in 
BFT protocols for the replicas (acceptors) from the quorums used for learning when a decision becomes
committed. 
More specifically,
we denote the quorum used for forming certificates (locking) by $\qlock$
and the quorum used for unlocking by $\qunlock$. 
We denote the quorum employed by clients for learning certificate uniqueness by
$\qunique$, and the quorum used for learning commit safety by $\qcommit$. 
In other words, clients mandate $\qunique$ first-round votes and $\qcommit$ second-round votes in order
to commit a decision. 
Below, we outline a modified PBFT-like protocol that uses these different quorum sizes instead of a single quorum size $q$.
We then introduce a new definition, Flexible Byzantine Quorums, that
capture the requirements needed for these quorums to provide
safety and liveness.

\begin{figure}[h]
\centering
\begin{boxedminipage}{\columnwidth}
\begin{itemize}[itemsep=4pt,leftmargin=*]

\item[-] \textbf{Propose.} 
  A leader $L$ proposes a value $b$.

\item[-] \textbf{Vote 1.} 
  On receiving the first value $b$ for a view $v$, a replica votes for $b$ 
  if it is \emph{safe}, as determined by a locking mechanism described below.
  A set of $\qlock$ votes forms a certificate $\CertN^{\qlock}(b)$.

\item[-] \textbf{Vote 2.} 
  On collecting $\CertN^{\qlock}(b)$,
  a replica ``locks'' on $b$ and votes for $\CertN^{\qlock}(b)$.

\item[-] \textbf{Commit.} 
  On collecting $\qunique$ votes for $b$ and $\qcommit$ votes for $\CertN^{\qlock}(b)$, 
  a client learns that proposal $b$ becomes a committed decision. 

\end{itemize}
If a replica locks on a value $b$ in a view, then it votes only for $b$ in subsequent views 
unless it ``unlocks'' from $b$ by learning that $\qunlock$
 replicas are not locked on $b$.
\end{boxedminipage}
\end{figure}

\begin{description}[topsep=8pt,itemsep=4pt]
\item[Flexible quorum intersection (a) within a view.]
Contrary to PBFT, in \name, a pair of $\qlock$ certificates need not necessarily
intersect in an honest replica. 
Indeed, locking on a value does not preclude conflicting locks.
It only mandates that every $\qlock$ quorum
intersects with every $\qunique$ quorum at at least one honest replica. 
For safety, it is essential
that the fraction of faulty replicas is less than $\qlock+\qunique-1$.

\item[Flexible quorum intersection (b) across views.]
If a client commits a value $b$ in a view, 
$\qcommit$ replicas lock on $b$ in that view.
For an honest replica among them to unlock from $b$,
$\qunlock$ replicas need to claim they are not locked on $b$.
This property mandates that every $\qunlock$ quorum intersects
with every $\qcommit$ quorum at at least one honest replica.
Thus, for safety, it is essential that the fraction of faulty
replicas is less than $\qunlock + \qcommit - 1$. 

\item[Flexible quorum availability within a view.]
For liveness, Byzantine replicas cannot exceed 
{$1-\max(\qunique, \qcommit, \qlock, \qunlock)$} 
so that the aforementioned quorums can be formed at different stages of the protocol. 
\end{description}

Given the above analysis, \name ensures safety  
if the fraction of faulty replicas is less than
$\min(\qunique + \qlock - 1, \qcommit + \qunlock - 1)$,
and provides liveness if the fraction of Byzantine replicas 
is at most $1-\max(\qunique, \qcommit, \qlock, \qunlock)$.
It is optimal to use \emph{balanced quorum sizes}
where $\qlock = \qunlock$ and $\qunique = \qcommit$.
To see this, first note that we should make sure 
$\qunique + \qlock = \qcommit + \qunlock$;
otherwise, suppose the right-hand side is smaller, 
then setting $(\qcommit,\qunlock)$ to equal $(\qunique,\qlock)$
improves safety tolerance without affecting liveness tolerance.
Next, observe that if we have $\qunique + \qlock = \qcommit + \qunlock$
but $\qlock > \qunlock$ (and hence $\qunique < \qcommit$),
then once again setting $(\qcommit,\qunlock)$ to equal $(\qunique,\qlock)$
improves safety tolerance without affecting liveness tolerance.

Thus, in this paper, we set $\qlock = \qmin$ and $\qunique =
\qcommit = \q$. Since replicas use $\qmin$ votes to lock,
these votes can always be used by the clients to commit
$\qcommit$ quorums. 
Thus, $\q \geq \qmin$.
The Flexible Byzantine Quorum requirements collapse
into the following two conditions.
\begin{description}
\item[Flexible quorum intersection.]
The fraction of faulty replicas is $< \q + \qmin - 1$.

\item[Flexible quorum availability.]
The fraction of Byzantine replicas is $\leq 1-\q$.
\end{description}

\paragraph{Tolerating \psing faults.}
If all faults in the system are Byzantine faults, 
then the best parameter choice is $\q=\qmin \geq 2/3$
for $<1/3$ fault tolerance,
and Flexible Byzantine Quorums degenerate to basic Byzantine quorums.
However, in our model, \psing replicas are only interested in attacking safety but not liveness.
This allows us to tolerate  
$\q + \qmin - 1$ total faults (Byzantine plus \psing), which can be more than $1/3$. 
For example, if we set $\qmin = 0.7$ and $\q = 0.8$, 
then such a protocol can tolerate $0.2$ Byzantine faults plus
$0.3$ \psing faults.
We discuss the choice for $\qmin$ and $\q$ and their rationale in Section~\ref{sec:discussion}.

\paragraph{Separating client commit rules from the replica protocol.}
\label{sec:client-coexist-partial}
A key property of the Flexible Byzantine Quorum approach is that
it decouples the BFT protocol from 
client commit rules. 
The decoupling allows 
clients assuming different fault models
to utilize the same protocol.
In the above protocol, the propose and two voting steps 
are executed by the replicas and they are only parameterized by $\qmin$.
The commit step can be carried by different clients using different commit thresholds $\q$. 
Thus, a fixed $\qmin$ determines a possible set of clients with
varying commit rules (in terms of Byzantine and \psing adversaries).
Recall that a Byzantine adversary can behave arbitrarily and thus may not
provide liveness whereas an \psing adversary only
intends to attack safety but not liveness.
Thus, a client who believes that a large fraction of
the adversary may attempt to break safety, not progress, can choose a larger $\q$.
By doing so, it seeks stronger safety against dishonest replicas, while trading 
 liveness.
Conversely, a client that assumes that a large fraction of the adversary
attacks liveness must choose a smaller $\q$.


%% file: protocol.tex
\section{Flexible BFT Protocol}
\label{sec:protocol}

In this section, we combine the ideas presented in
Sections~\ref{sec:overv-synchr-flex} and \ref{sec:overview-async}
to obtain a final protocol that supports both types of clients.
A client can either assume partial synchrony, with freedom to choose $\q$ as described
in the previous section, or assume synchrony with its own choice of
$\Delta$, as described in Section~\ref{sec:overv-synchr-flex}. 
Replicas execute a protocol at the network speed
with a parameter $\qmin$.
We first give the protocol executed by the replicas 
and then discuss how clients commit depending on their assumptions. 
Moreover, inspired by Casper~\cite{DBLP:journals/corr/abs-1710-09437} and
HotStuff~\cite{yin2018hotstuff}, we show a protocol where the
rounds of voting can be pipelined.

\subsection{Notation}
Before describing the protocol, we will first define some
data structures and terminologies that will aid presentation.

\paragraph{Block format.} 
The pipelined protocol forms a chain of values.
We use the term \emph{block} to refer to each value in the chain.
We refer to a block's position in the chain as its \emph{height}. 
A block $\block_k$ at height $k$ has the following format
$$\block_k := (b_k, h_{k-1})$$ 
where $b_k$ denotes a proposed value at height $k$ 
and $h_{k-1} := H(\block_{k-1})$ is a hash digest of the predecessor block.
The first block $\block_1=(b_1, \bot)$ has no predecessor.
Every subsequent block $\block_k$ must specify a predecessor block $\block_{k-1}$ by including a hash of it.
We say a block is \emph{valid} if 
(i) its predecessor is valid or $\bot$, and 
(ii) its proposed value meets application-level validity conditions and is consistent with its chain of ancestors (e.g., does not double spend a transaction in one of its ancestor blocks). 

\paragraph{Block extension and equivocation.}
We say $\block_l$ \emph{extends} $\block_k$, 
if $\block_k$ is an ancestor of $\block_l$ ($l>k$). 
We say two blocks $\block_l$ and $\block'_{l'}$ \emph{equivocate} one another 
if they are not equal and do not extend one another.

\paragraph{Certificates and certified blocks.} 
In the protocol, replicas vote for blocks by signing them.
We use $\Cert_{\ViewNumber}{(\block_{k})}$ to denote a set of signatures 
on $h_{k} = H(\block_{k})$ by $\qmin$ replicas in view $\ViewNumber$. 
$\qmin$ is a parameter fixed for the protocol instance. 
We call $\Cert_{\ViewNumber}{(\block_{k})}$ a certificate for $\block_{k}$ from view $v$.
Certified blocks are ranked 
first by the views in which they are certified and then by their heights.
In other words, a block $\block_k$ certified in view $\ViewNumber$ 
is ranked \emph{higher} than a block $\block_{k'}$ certified in view $\ViewNumber'$
if either (i) $\ViewNumber > \ViewNumber'$ or (ii) $\ViewNumber = \ViewNumber'$ and $k>k'$. 

\paragraph{Locked blocks.} 
At any time, a replica locks the highest certified block to its knowledge.
During the protocol execution, each replica keeps track of all signatures for all blocks
and keeps updating its locked block.
Looking ahead, the notion of locked block will be
used to guard the safety of a client commit.

\subsection{Replica Protocol}

The replica protocol progresses in a view-by-view fashion. Each view has a
designated leader who is responsible for driving consensus on a
sequence of blocks. Leaders can be chosen
statically, e.g., round robin, 
or randomly using more sophisticated techniques~\cite{CKS05,Algorand}. 
In our description, we assume a round robin selection of leaders,
i.e.,  ($\ViewNumber$ {\sf mod} $n$) is the leader of view $\ViewNumber$. 

At a high level, the protocol does the following: 
The leader proposes a block to all replicas. 
The replicas vote on it if safe to do so. 
The block becomes certified once $\qmin$ replicas vote on it. 
The leader will then propose another block extending the previous one, 
chaining blocks one after another at increasing heights.
Unlike regular consensus protocols
where replicas determine when a block is committed, in
\name, replicas only certify blocks while committing is
offloaded to the clients.
If at any time replicas detect malicious leader behavior or lack of
progress in a view, they blame the leader 
and engage in a view change protocol 
to replace the leader and move to the next view. 
The new leader collects a status from different replicas and continues to propose
blocks based on this status.
We explain the steady state and view change protocols in more detail below.

\begin{figure*}[tb]
\centering
\begin{boxedminipage}{\textwidth}

Let $\ViewNumber$ be the current view number and 
replica $L$ be the leader in this view. 
Perform the following steps in an iteration.

\begin{enumerate}[topsep=8pt,itemsep=8pt,leftmargin=*]
\setlength\itemsep{0.5em}
\item \textbf{Propose. } \label{step:propose} 
\Comment{Executed by the leader of view $\ViewNumber$}

The leader $L$ broadcasts $\sig{\Propose, \block_k, \ViewNumber, \Cert_{\ViewNumber'}(\block_{k-1}), \StatusCert}_L$.
Here, $\block_k := (b_k, h_{k-1})$ is the newly proposed block
and it should extend the highest certified block known to $L$.
In the steady state, an honest leader $L$ would extend the
previous block it proposed,
in which case $\ViewNumber'=\ViewNumber$ and $\StatusCert = \bot$. 
Immediately after a view change, 
$L$ determines the highest certified block from
the status $\StatusCert$ received during the view change. 

\item \label{step:vote} \textbf{Vote.} 
\Comment{Executed by all replicas}

When a replica $R$ receives a valid proposal 
$\sig{\Propose, \block_k, \ViewNumber, \Cert_{\ViewNumber'}(\block_{k-1}), \StatusCert}_L$
from the leader $L$, $R$ broadcasts the proposal and a vote $\sig{\Vote, \block_k, \ViewNumber}_R$ 
if (i) the proposal is the first one in view $\ViewNumber$,
and it extends the highest certified block in $\StatusCert$, 
or (ii) the proposal extends the last proposed block in the view. 


In addition, replica $R$ records the following based on the messages it receives.
\begin{itemize}[topsep=8pt,itemsep=4pt]
\item[-] $R$ keeps track of the number of
  votes received for this block in this view as $q_{\block_{k},\ViewNumber}$. 
\item[-] If block $\block_{k-1}$ has been proposed in view
  $\ViewNumber$, $R$ marks $\block_{k-1}$ as a locked block and
  records the locked time as $\TimeAccept_{k-1,\ViewNumber}$.
\item[-] If a block equivocating $\block_{k-1}$
  is proposed by $L$ in view $\ViewNumber$ (possibly received through a vote),
  $R$ records the time $\TimeEquiv_{k-1,\ViewNumber}$ at
  which the equivocating block is received. 
\end{itemize}

The replica then enters the next iteration.
If the replica observes no progress or equivocating blocks in the same view $\ViewNumber$, 
it stops voting in view $\ViewNumber$ 
and sends $\sig{\Blame,\ViewNumber}_r$ message to all replicas. 
\end{enumerate}

\end{boxedminipage}
\caption{Flexible BFT steady state protocol.}
\label{fig:steady}
\end{figure*}

\paragraph{Steady state protocol.}
The steady state protocol is described in Figure~\ref{fig:steady}.
In the steady state, there is a unique leader who, in an iteration,
proposes a block, waits for votes from $\qmin$
replicas and moves to the next iteration.
In the steady state, an honest leader always extends the previous block it proposed.
Immediately after a view change, 
since the previous leaders could have been
Byzantine and may have proposed equivocating blocks, 
the new leader needs to determine a safe block to propose. 
It does so by collecting a status of
locked blocks from $\qmin$ replicas denoted by $\StatusCert$
(described in the view change protocol).

For a replica $R$ in the steady state, on receiving a proposal for block
$\block_{k}$, a replica votes for it if 
it extends the previous proposed block in the view or 
if it extends the highest certified block in $\StatusCert$.
Replica $R$ can potentially receive blocks out of
order and thus receive $\block_{k}$ before its ancestor blocks. 
In this case, replica $R$ waits until it receives the ancestor blocks, 
verifies the validity of those blocks and $\block_k$ before voting for $\block_k$.
In addition, replica $R$ records the following to aid a client commit:
\begin{itemize}
\item[-] \textbf{Number of votes.} It records the number of votes received
  for $\block_{k}$ in view $\ViewNumber$ as $q_{\block_{k},
    \ViewNumber}$. Observe that votes are broadcast by 
  all replicas and the number of votes for a block can be greater
  than $\qmin$. $q_{\block_{k}, \ViewNumber}$ will be updated each time the
  replica hears about a new vote in view $\ViewNumber$.
\item[-] \textbf{Lock time.} If $\block_{k-1}$ was proposed in the
  same view $\ViewNumber$, it locks $\block_{k-1}$ and records the locked time as
  $\TimeAccept_{k-1, \ViewNumber}$.
\item[-] \textbf{Equivocation time.} If the replica ever observes an
  equivocating block at height $k$ in view $\ViewNumber$ through a
  proposal or vote, 
  it stores the time of equivocation as $\TimeEquiv_{k,\ViewNumber}$.
\end{itemize}
Looking ahead, the locked time $\TimeAccept_{k-1,\ViewNumber}$ and
equivocation time $\TimeEquiv_{k-1,\ViewNumber}$ will
be used by clients with synchrony assumptions to commit,
and the number of votes $q_{\block_{k}, \ViewNumber}$
will be used by clients with partial-synchrony assumptions to commit.

\paragraph{Leader monitoring.}
If a replica detects a lack of progress in view $\ViewNumber$ or
observes malicious leader behavior
such as more than one height-$k$ blocks in the same view,
it blames the leader of view $\ViewNumber$
by broadcasting a $\sig{\Blame, \ViewNumber}$ message. It quits
view $\ViewNumber$ and stops voting and broadcasting blocks in
view $\ViewNumber$. 
To determine lack of progress, the replicas may simply guess a time bound for
message arrival or use increasing timeouts for each view~\cite{castro1999practical}.

\paragraph{View change.}
The view change protocol is described in Figure~\ref{fig:vc}.
If a replica gathers $\qmin$ $\sig{\Blame, \ViewNumber}$
messages from distinct replicas,
it forwards them to all other replicas and enters a new view
$\ViewNumber + 1$ (Step~\ref{step:new_view}). It records the time
at which it received the blame certificate as $\TimeBlame_{\ViewNumber}$.
Upon entering a new view,
a replica reports to the leader of the new view $L'$ its locked block 
and transitions to the steady state (Step~\ref{step:status}). 
$\qmin$ status messages form the status $\StatusCert$.
The first block $L'$ proposes in the new view 
should extend the highest certified block among these $\qmin$ status messages.

\begin{figure*}[htbp]
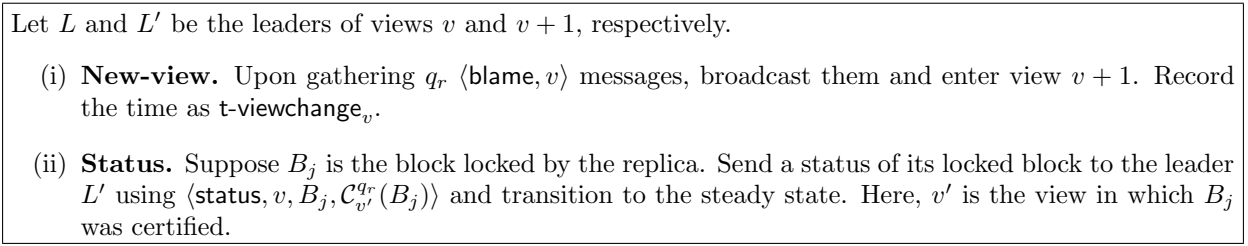

\centering
\begin{boxedminipage}{\textwidth}

Let $L$ and $L'$ be the leaders of views $\ViewNumber$ and $\ViewNumber+1$, respectively.
\begin{enumerate}[topsep=8pt,itemsep=8pt,leftmargin=*,label=(\roman*)]
\setlength\itemsep{0.5em}
\item \label{step:new_view} \textbf{New-view.}
Upon gathering $\qmin$ $\sig{\Blame, \ViewNumber}$ messages,
broadcast them and enter view $\ViewNumber+1$. Record the time as $\TimeBlame_{\ViewNumber}$.

\item \label{step:status}  \textbf{Status.}
Suppose $\block_j$ is the block locked by the replica. Send a
status of  its locked block to the leader $L'$ using
$\sig{\Status, v, \block_j, \Cert_{v'}(\block_j)}$ and transition
to the steady state. Here, $v'$ is the view in which $\block_j$
was certified.

\end{enumerate}

\end{boxedminipage}
\caption{Flexible BFT view change protocol.}
\label{fig:vc}
\end{figure*}

\subsection{Client Commit Rules}
\label{sec:commit-rules}

\begin{figure*}[htbp]
  \centering
  \begin{boxedminipage}{\textwidth}
    \begin{enumerate}[itemsep=8pt,leftmargin=0.2cm,label=]
      
      \item \textbf{(CR1) Partially-synchronous commit.} A
        block $\block_k$ is committed under the partially synchronous rule with parameter $\q$  
        iff there exist $l \geq k$ and $v$ such that
        \begin{enumerate}[topsep=8pt,itemsep=4pt]
				\item $\Cert_{\ViewNumber}(\block_l)$ and $\Cert_{\ViewNumber}(\block_{l+1})$ 
						exist where $\block_{l+1}$ extends $\block_l$ and $\block_k$ (if $l = k$, $\block_l = \block_k$).
        \item $q_{\block_{l}, \ViewNumber} \geq \q$ and $q_{\block_{l+1}, \ViewNumber} \geq \q$.
        \end{enumerate}
        
		\item \textbf{(CR2) Synchronous commit.}
        A block $\block_k$ is committed assuming $\Delta-$synchrony 
		iff the following holds for $\qmin$ replicas.
		There exist $l \geq k$ and $v$ (possibly
                different across replicas) such that,
        \begin{enumerate}[topsep=8pt,itemsep=4pt]
        \item $\Cert_{\ViewNumber}(\block_l)$ exists where $\block_l$ extends $\block_k$ (if $l = k$, $\block_l = \block_k$).
        \item An undisturbed-$2\Delta$ period is observed after
          $\block_{l+1}$ is obtained, i.e., no equivocating block
          or view change of view $v$ were
           observed before $2\Delta$ time after $\block_{l+1}$
           was obtained, i.e.,
        $$\min(\TimeCurrent, \TimeEquiv_{l,\ViewNumber},
        \TimeBlame_\ViewNumber) - \TimeAccept_{l,\ViewNumber}
        \geq 2\Delta$$
        \end{enumerate}
    \end{enumerate}
  \end{boxedminipage}
  \caption{Flexible BFT commit rules}
  \label{fig:commit-rules}
\end{figure*}

As mentioned in the introduction, \name supports 
clients with different assumptions. 
Clients in \name learn the
state of the protocol from the replicas and based on their
own assumptions determine whether a block has been committed.
Broadly, we supports two types of clients:
those who believe in synchrony and those who believe in partial synchrony.

\subsubsection{Clients with Partial-Synchrony Assumptions (CR1)}
A client with partial-synchrony assumptions deduces whether a block has been committed
by based on the number of votes received by a block. 
A block $\block_l$ (together with its ancestors) is committed with parameter $\q$
iff $\block_l$ and its immediate
successor both receive $\geq \q$ votes in the same view. 

\paragraph{Safety of CR1.}
A CR1 commit based on $\q$ votes is safe 
against $<\q+\qmin-1$ faulty replicas (Byzantine plus \psing).
Observe that if $\block_l$ gets $\q$ votes in view $\ViewNumber$,  
due to flexible quorum intersection, 
a conflicting block cannot be certified in view $\ViewNumber$, 
unless $\geq \q+\qmin-1$ replicas are faulty. 
Moreover, $\block_{l+1}$ extending $\block_l$ 
has also received $\q$ votes in view $\ViewNumber$. 
Thus, $\q$ replicas lock block $\block_l$ in view $\ViewNumber$.
In subsequent views, honest replicas that have locked 
$\block_l$ will only vote for a block 
that equals or extends $\block_l$ unless they unlock. 
However, due to flexible quorum intersection, 
they will not unlock 
unless $\geq \q+\qmin-1$ replicas are faulty. 
Proof of Lemma~\ref{lemma:unique-cert} formalizes this argument. 

\subsubsection{Client with Synchrony Assumptions (CR2)}
Intuitively, a CR2 commit involves $\qmin$ replicas collectively stating that 
no ``bad event'' happens within ``sufficient time'' in a view.
Here, a bad event refers to either leader equivocation or view change
(the latter indicates sufficient replicas believe leader is
faulty) and the ``sufficient time'' is $2\Delta$; 
where $\Delta$ is a synchrony bound chosen by the client. 
More formally, a replica states that a synchronous commit for block $\block_k$
for a given parameter $\Delta$ (set by a client) is satisfied iff the following holds.
There exists $\block_{l+1}$ that extends $\block_l$ and $\block_k$,
and the replica observes an undisturbed-$2\Delta$ period after
obtaining $\block_{l+1}$ during which (i) no equivocating block
is observed, and (ii) no blame certificate/view
change certificate for view $v$ was obtained, i.e., 
$$\min(\TimeCurrent, \TimeEquiv_{l,\ViewNumber},
        \TimeBlame_\ViewNumber) - \TimeAccept_{l,\ViewNumber}
        \geq 2\Delta$$
where $\TimeEquiv_{l,\ViewNumber}$ denotes the time 
equivocation for $\block_l$ in view $\ViewNumber$ was
observed ($\infty$ if no equivocation), $\TimeBlame_\ViewNumber$
denotes the time at which view change happened from view $\ViewNumber$
to $\ViewNumber + 1$ ($\infty$ if no view change has happened
yet), and $\TimeAccept_{l,\ViewNumber}$ denotes the time at which 
$\block_l$ was locked (or $\block_{l+1}$ was proposed) in
view $\ViewNumber$.
Note that the client does not require the $\qmin$ fraction of replicas 
to report the same height $l$ or view $\ViewNumber$.

\paragraph{Safety of CR2.}
A client believing in synchrony assumes that all messages between replicas
arrive within $\Delta$ time after they were sent. 
If the client's chosen $\Delta$ is a correct upper bound on message delay, 
then a CR2 commit is safe against $\qmin$ faulty replicas (Byzantine plus \psing),
as we explain below.
If less than $\qmin$ replicas are faulty, at least one honest replica reported 
an \emph{undisturbed-$2\Delta$} period.
Let us call this honest replica $h$ and 
analyze the situation from $h$'s perspective
to explain why an undisturbed $2\Delta$ period ensures safety.
Observe that replicas in \name forward the proposal when voting. 
If $\Delta$-synchrony holds, every other honest replica 
learns about the proposal $\block_l$ at most $\Delta$ time after $h$ learns about it.
If any honest replica voted for a conflicting block 
or quit view $\ViewNumber$, $h$ would have known within $2\Delta$ time. 

\input{proof}

\subsection{Efficiency}

\paragraph{Latency.}
Clients with a synchrony assumption incur a latency of
$2\Delta$ plus a few network speed rounds. In terms of the
maximum network delay $\Delta$, this matches the 
state-of-the-art synchronous
protocols~\cite{abraham2019sync}. The distinction though is that
$\Delta$ now depends on the client assumption and 
hence different clients may commit with different latencies
Clients with partial-synchrony assumptions incur a latency of two
rounds of voting; this matches PBFT~\cite{castro1999practical}.

\paragraph{Communication.}
Every vote and new-view messages are broadcast to all replicas,
incurring $O(n^2)$ communication messages. This is the same
complexity of PBFT~\cite{castro1999practical} and Sync
HotStuff~\cite{abraham2019sync}.


%% file: proof.tex
\subsection{Safety and Liveness}
\label{sec:proof}

We introduce the notion of \emph{direct} and \emph{indirect} commit to aid the proofs.
We say a block is committed \emph{directly} under \textbf{CR1} if
the block and its  immediate successor both get $\q$ votes in the same view.
We say a block is committed \emph{directly} under \textbf{CR2} if some honest replica
reports an undisturbed-$2\Delta$ period after its successor block
was obtained.
We say a block is committed \emph{indirectly} if neither condition applies to it but
it is committed as a result of a block extending it being committed directly.
We remark that the direct commit notion, especially for \textbf{CR2}, is merely a proof technique.
A client cannot tell whether a replica is honest,
and thus has no way of knowing whether a block is directly committed under \textbf{CR2}. 

\begin{lemma}
If a client directly commits a block $\block_l$ in view $v$ using a correct commit rule,
then a certified block that ranks no lower than $\Cert_{v}(\block_{l})$ must equal or extend $\block_l$. 
\label{lemma:unique-cert}
\end{lemma}

\begin{proof}
To elaborate on the lemma, 
a certified block $\Cert_{v'}(\block'_{l'})$ ranks no lower than $\Cert_{v}(\block_{l})$
if either (i) $v'=v$ and $l' \geq l$, or (ii) $v'>v$.
We need to show that if $\block_l$ is directly committed, 
then any certified block that ranks no lower either equals or extends $\block_l$.
We consider the two commit rules separately. 
For both commit rules, we will use induction on $v'$ to prove the lemma.

\bigskip
For $\textbf{CR1}$ with parameter $\q$ to be correct, 
flexible quorum intersection needs to hold, i.e.,
the fraction of faulty replicas must be less than $\q+\qmin-1$.
$\block_l$ being directly committed under $\textbf{CR1}$ with parameter $\q$ implies that 
there are $\q$ votes in view $v$ for $\block_l$ and $\block_{l+1}$ where $\block_{l+1}$ extends $\block_l$. 

For the base case, a block $\block'_{l'}$ with $l'\geq l$ that does not extend $\block_l$ cannot get certified in view $v$,
because that would require $\q+\qmin-1$ replicas to vote for two equivocating blocks in view $v$.

Next, we show the inductive step.
Note that $\q$ replicas voted for $\block_{l+1}$ in view $v$,
which contains $\Cert_{v}(\block_{l})$.
Thus, they lock $\block_l$ or a block extending $\block_l$ by the end of view $v$.
Due to the inductive hypothesis,
any certified block that ranks equally or higher from view $v$ up to view $v'$ 
either equals or extends $\block_l$. 
Thus, by the end of view $v'$,
those $\q$ replicas still lock $\block_l$ or a block extending $\block_l$.
Since the total fraction of faults is less than $\q+\qmin-1$,
the status $\StatusCert$ shown by the leader of view $v'+1$ 
must include a certificate for $\block_l$ or a block extending it;
moreover, any certificate that ranks equal to or higher than $\Cert_{v}(\block_{l})$
is for a block that equals or extends $\block_l$. 
Thus, only a block that equals or extends $\block_l$ can gather votes from those $\q$ replicas in view $v'+1$
and only a block that equals or extends $\block_l$ can get certified in view $v'+1$.

\bigskip
For $\textbf{CR2}$ with synchrony bound $\Delta$ to be correct, 
$\Delta$ must be an upper bound on worst case message delay 
and the fraction of faulty replicas is less than $\qmin$. 
$\block_l$ being directly committed under $\textbf{CR2}$ with $\Delta$-synchrony implies that
at least one honest replica voted for $\block_{l+1}$ extending $\block_l$ in view $v$,
and did not hear an equivocating block or view change within $2\Delta$ time after that.
Call this replica $h$.
Suppose $h$ voted for $\block_{l+1}$ extending $\block_l$ in view $v$ at time $t$,
and did not hear an equivocating block or view change by time $t+2\Delta$. 

We first show the base case: a block $\block'_{l'}$ with $l'\geq l$ certified in view $v$ must equal or extend $\block_l$.
Observe that if $\block'_{l'}$ with $l'\geq l$ does not equal or extend $\block_l$,
then it equivocates $\block_l$.
No honest replica voted for $\block'_{l'}$ before time $t+\Delta$,
because otherwise $h$ would have received the vote for $\block'_{l'}$ by time $t+2\Delta$,
No honest replica would vote for $\block'_{l'}$ after time $t+\Delta$ either,
because by then they would have received (from $h$) and voted for $\block_l$.
Thus, $\block'_{l'}$ cannot get certified in view $v$.

We then show the inductive step.
Because $h$ did not hear view change by time $t+2\Delta$,
all honest replicas are still in view $v$ by time $t+\Delta$,
which means they all receive $\block_{l+1}$ from $h$ by the end of view $v$.
Thus, they lock $\block_l$ or a block extending $\block_l$ by the end of view $v$.
Due to the inductive hypothesis,
any certified block that ranks equally or higher from view $v$ up to view $v'$ 
either equals or extends $\block_l$. 
Thus, by the end of view $v'$,
all honest replicas still lock $\block_l$ or a block extending $\block_l$.
Since the total fraction of faults is less than $\qmin$,
the status $\StatusCert$ shown by the leader of view $v'+1$ 
must include a certificate for $\block_l$ or a block extending it;
moreover, any certificate that ranks equal to or higher than $\Cert_{v}(\block_{l})$
is for a block that equals or extends $\block_l$. 
Thus, only a block that equals or extends $\block_l$ can gather honest votes in view $v'+1$
and only a block that equals or extends $\block_l$ can get certified in view $v'+1$.
\end{proof}

\begin{theorem}[Safety]
Two clients with correct commit rules commit the same block $\block_k$ for each height $k$.
\label{thm:safety}
\end{theorem}

\begin{proof}
Suppose for contradiction that two distinct blocks 
$\block_k$ and $\block'_k$ are committed at height $k$.
Suppose $\block_k$ is committed as a result of $\block_{l}$ being directly committed in view $v$
and $\block'_k$ is committed as a result of $\block'_{l'}$ being directly committed in view $v'$.
This implies $\block_l$ is or extends $\block_k$;
similarly, $\block'_{l'}$ is or extends $\block'_k$.
Without loss of generality, assume $v \leq v'$.
If $v=v'$, further assume $l \leq l'$ without loss of generality.
By Lemma~\ref{lemma:unique-cert},
the certified block $\Cert_{v'}(\block'_{l'})$ must equal or extend $\block_l$.
Thus, $\block'_k=\block_k$.
\end{proof}

\begin{theorem}[Liveness]
If all clients have correct commit rules, 
they all keep committing new blocks. 
\label{thm:liveness}
\end{theorem}

\begin{proof}
By the definition of \psing faults,
if they cannot violate safety, they will preserve liveness.
Theorem~\ref{thm:safety} shows that if all clients have correct commit rules,
then safety is guaranteed \emph{even if \psing replicas behave arbitrarily}.
Thus, once we proved safety, we can treat \psing replicas 
as honest when proving liveness. 

Observe that a correct commit rule tolerates at most $1-\qmin$ Byzantine faults.
If a Byzantine leader prevents liveness, 
there will be $\qmin$ blame messages against it,
and a view change will ensue to replace the leader. 
Eventually, a non-Byzantine (honest or \psing) replica becomes the leader 
and drives consensus in new heights.
If replicas use increasing timeouts,
eventually, all non-Byzantine replicas stay in the same view for sufficiently long.
When both conditions occur, 
if a client's commit rule is correct (either \textbf{CR1} and \textbf{CR2}),
due to quorum availability,
it will receive enough votes in the same view to commit. 
\end{proof}


%% file: discussion.tex
\section{Discussion}
\label{sec:discussion}

As we have seen, three parameters 
$\qmin$, $\q$, and $\Delta$ determine the protocol. 
$\qmin$ is the only parameter for the replicas
and is picked by the service administrator.
The choice of $\qmin$
determines a set of client assumptions that can be supported. 
$\q$ and $\Delta$ are chosen by clients to commit blocks.
In this section, we first discuss the client assumptions supported by a
given $\qmin$ and then discuss the trade-offs 
between different choices of $\qmin$.

\subsection{Client Assumptions Supported by $\qmin$}
\label{sec:client-beliefs-given}

\begin{figure}[tbp]
  \centering
    \includegraphics[width=0.45\textwidth]{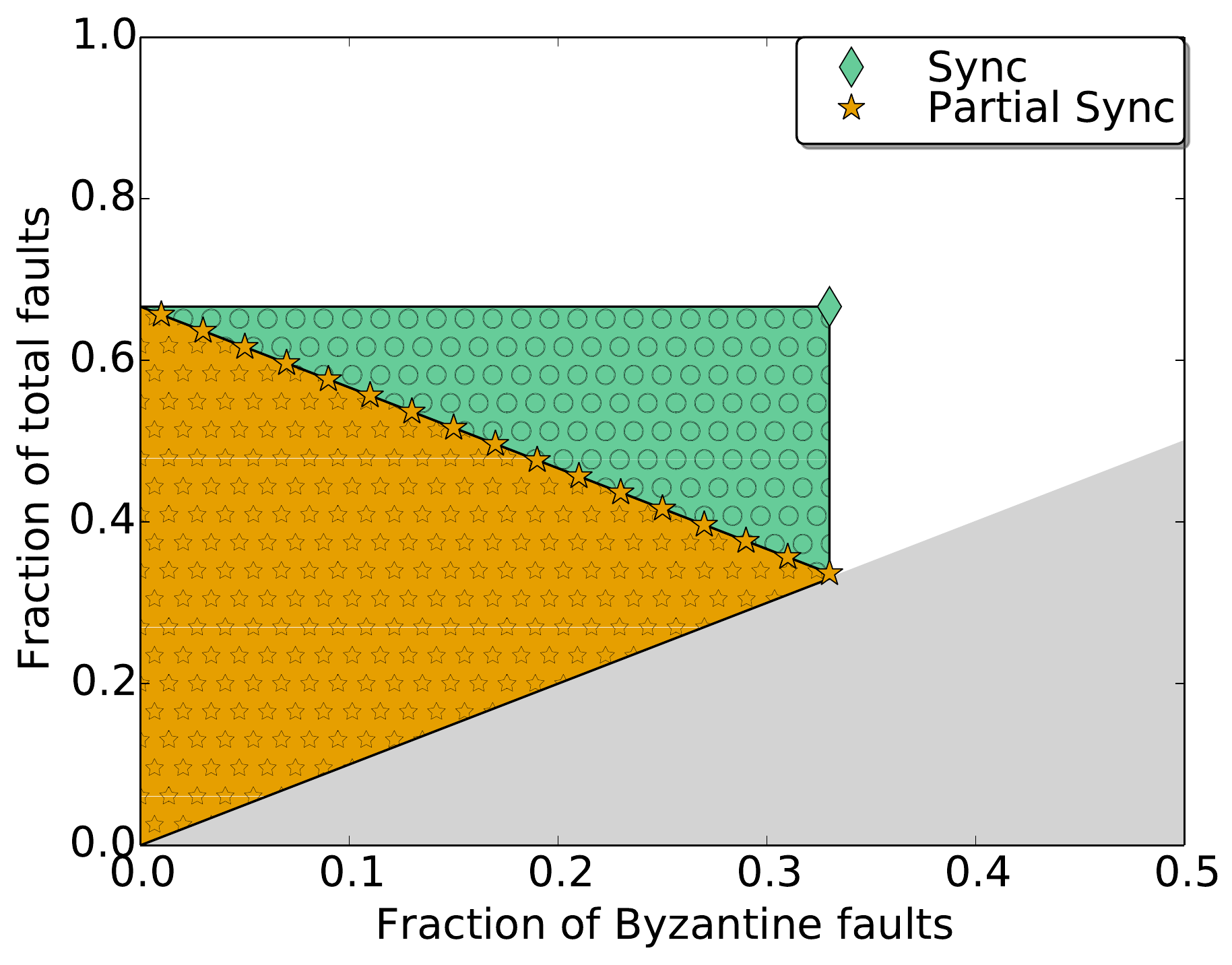}
    \caption{\textbf{Clients supported for $\qmin = 2/3$.}}
    \label{fig:clients-supported}
\end{figure}

Figure~\ref{fig:clients-supported} represents
the clients supported at $\qmin = 2/3$. 
The x-axis represents Byzantine faults
and the y-axis represents total faults (Byzantine plus \psing). 
Each point on this graph 
represents a client fault assumption as a pair: (Byzantine faults, total faults).
The shaded gray area indicates an ``invalid area'' since we
cannot have fewer total faults than Byzantine faults.
A missing dimension
in this figure is the choice of $\Delta$. 
Thus, the synchrony guarantee shown in this figure is for clients
that choose a correct synchrony bound. 

Clients with partial-synchrony assumptions
can get fault tolerance on (or below) the starred orange line. 
The right most point on the line is $(1/3, 1/3)$, i.e., we
tolerate less than a third of Byzantine replicas and no additional
\psing replicas.  
This is the setting of existing partially synchronous consensus
 protocols~\cite{DLS88,castro1999practical,yin2018hotstuff}. 
\name generalizes these protocols by giving clients the option of
moving up-left along the line, 
i.e., tolerating fewer Byzantine and more total faults.
By choosing $\q>\qmin$, a client tolerates 
$< \q+\qmin-1$ total faults for safety 
and $\leq 1-\q$ Byzantine faults for liveness. 
In other words, as a client moves left,  
for every additional vote it requires,
it tolerates one fewer Byzantine fault and gains overall
one higher total number of faults (i.e., two more \psing faults).
The left most point on this line $(0, 2/3)$ tolerating no 
Byzantine replicas and the highest fraction of \psing replicas. 

Moreover, for clients who believe in synchrony,
if their $\Delta$ assumption is correct,
they enjoy 1/3 Byzantine tolerance and 2/3 total tolerance
represented by the green diamond.
This is because synchronous commit rules are not parameterized by
the number of votes received. 

\paragraph{How do clients pick their commit rules?}
In Figure~\ref{fig:clients-supported}, the shaded starred orange
portion of the plot represent fault tolerance provided 
by the partially synchronous commit rule (CR1). 
Specifically, setting $\q$ to the total fault 
fraction yields the necessary commit rule. On the other hand, if
a client's required fault tolerance lies in the circled green portion
of the plot, then the synchronous commit rule (CR2) with an
appropriate $\Delta$ picked by the client yields the necessary
commit rule. Finally, if a client's target fault tolerance corresponds to the
white region of the plot, then 
it is not achievable with this $\qmin$.

\paragraph{Clients with incorrect assumptions and recovery.}
If a client has incorrect assumption with respect to the fault
threshold or synchrony parameter $\Delta$, then it can lose
safety or liveness. If a client believing in synchrony
picks too small a $\Delta$ and commits a value $b$, 
it is possible that a conflicting
value $b'$ may also be certified. Replicas may
choose to extend the branch containing $b'$, effectively
reverting $b$ and causing a safety violation. Whenever a client
detects such a safety violation, it may need to revert
some of its commits and increase $\Delta$ to recover.

For a client with partial-synchrony assumption, if it loses safety, 
it can update its fault model to move left along
the orange starred line, i.e., tolerate higher total faults but fewer
Byzantine. On the other hand, if it observes no progress as its
threshold $\q$ is not met, then it moves towards the
right. However, if the true fault model is in the circled green
region in Figure~\ref{fig:clients-supported}, then the client
cannot find a partially synchronous commit rule that is both safe and live
and eventually has to switch to using a synchronous commit rule.

Recall that the goal of \psing replicas is to attack 
safety. Thus, clients with
incorrect assumptions may be exploited by
\psing replicas for their own gain (e.g., by double-spending). 
When a client updates to a correct assumption 
and recovers from unsafe commits, their subsequent commits
would be safe and final. This is remotely analogous to Bitcoin 
-- if a client commits to a transaction when it is a few blocks deep and a
powerful adversary succeeds in creating an alternative longer fork, the
commit is reverted. 

\subsection{Comparing Different $\qmin$ Choices}
\label{sec:comp-diff-qmins}
\begin{figure}[tbp]
  \centering
    \includegraphics[width=0.45\textwidth]{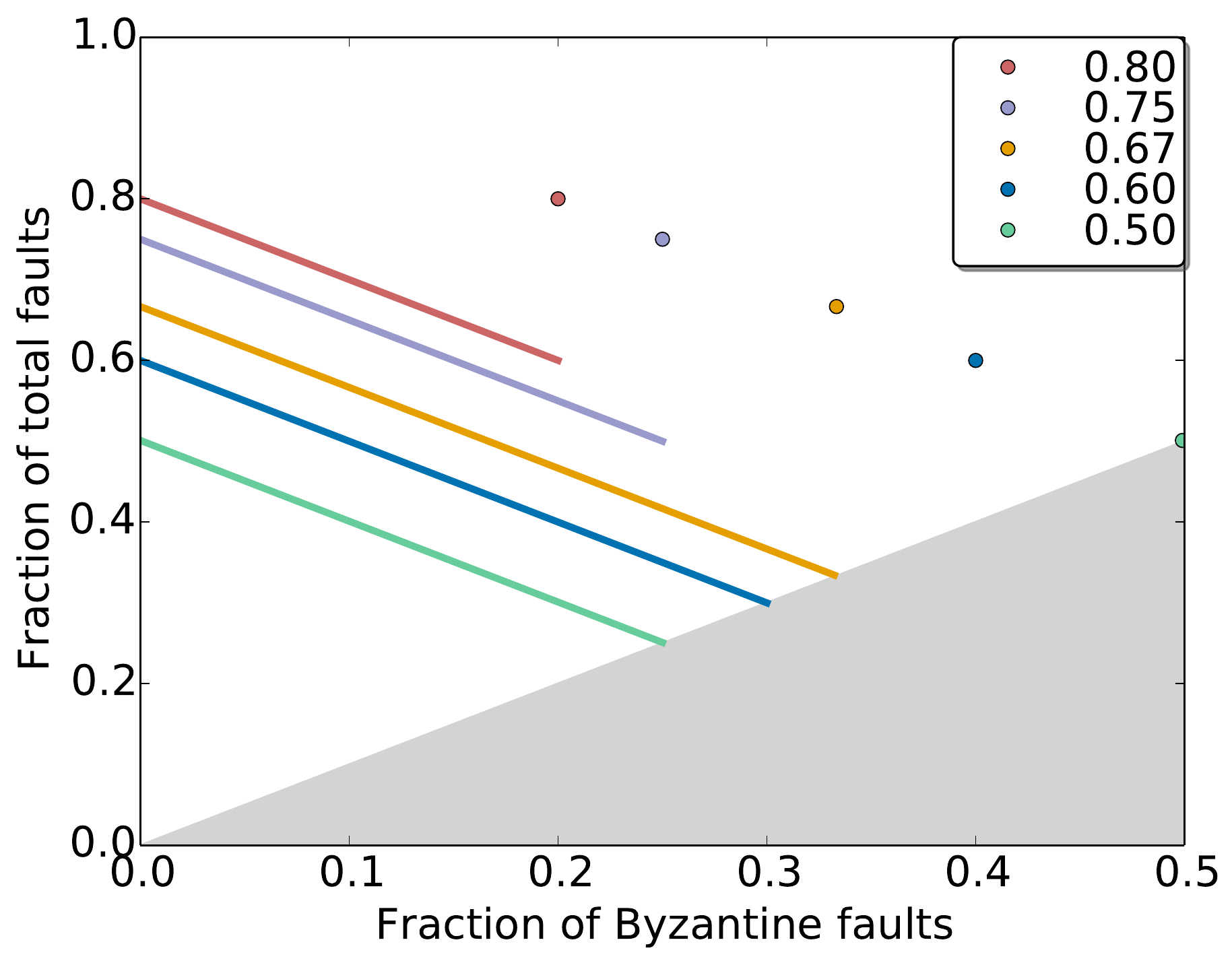}
    \caption{Clients supported by \name at different
      $\qmin$'s.  The legend represents the different $\qmin$ values.}
    \label{fig:varying-qmin}
\end{figure}
We now look at the service administrator's choice at picking
$\qmin$. 
In general, the service administrator's goal is to tolerate a large number of Byzantine
and \psing faults, i.e., move towards top and/or right of the figure.
Figure~\ref{fig:varying-qmin} shows the trade-offs
in terms of clients supported by different $\qmin$ values in \name.

First, it can be observed that for clients with partial-synchrony assumptions, 
$\qmin \geq 2/3$ dominates $\qmin < 2/3$. 
Observe that the fraction of Byzantine
replicas $(B)$ are bounded by $B < \q+\qmin-1$ and $B \leq
1-\q$, so $B \leq \qmin/2$. 
Thus, as $\qmin$ decreases, Byzantine fault tolerance decreases.
Moreover, since the total fault tolerance is 
$\q + \qmin - 1$, a lower $\qmin$ also tolerates a
smaller fraction of total faults for a fixed $\q$.

For $\qmin \geq 2/3$ or for clients believing in synchrony, 
no value of $\qmin$ is Pareto optimal. 
For clients with partial-synchrony assumptions, as $\qmin$ increases, 
the total fault tolerance for safety increases.
But since $\q \geq \qmin$, we have $B \leq 1 - \qmin$, and
hence the Byzantine tolerance for liveness decreases.
For clients believing in synchrony, the total fault tolerance
for safety is $< \qmin$ and the Byzantine fault tolerance for
liveness is $\geq 1-\qmin$. 
In both cases, the choice of $\qmin$ represents a safety-liveness
trade-off.

\subsection{Separating \Pleasing Resilience from Diversity}

So far, we presented the \name techniques and protocols 
to simultaneously support diverse client support and stronger \psing fault tolerance.
Indeed, we believe both properties are desirable and they strengthen each other.
But we remark that these two properties can be provided separately.

It is relatively straightforward to provide stronger fault tolerance in the \psing model in a classic uniform setting. 
For example, under partial-synchrony, one can simply use a larger quorum in PBFT (without the $\qmin$/$q$ replica/client quorum separation).
But we note that a higher total (\psing plus Byzantine) tolerance comes at the price of a lower Byzantine tolerance.
In a uniform setting, this means \emph{all} clients have to sacrifice some Byzantine tolerance.
In the diverse setting, \name gives clients the freedom to choose the fault assumption they believe in,
and a client can choose the classic Byzantine fault model.

On the flip side, if one hopes to support diverse clients in the classic Byzantine fault (no \psing faults),
the ``dimension of diversity'' reduces.
One example is the network speed replica protocol in Section~\ref{sec:overv-synchr-flex},
which supports clients that believe in different synchrony bounds.
That protocol can be further extended to support clients with a (uniform) partial-synchrony assumption.
Clients with partial-synchrony assumption are uniform since
we have not identified any type of ``diversity'' outside \psing faults for them.


%% file: related.tex
\section{Related Work}
\label{sec:related-work}

\begin{figure}[htbp]
  \centering
  \begin{subfigure}{0.45\textwidth}
    \includegraphics[width=1\linewidth]{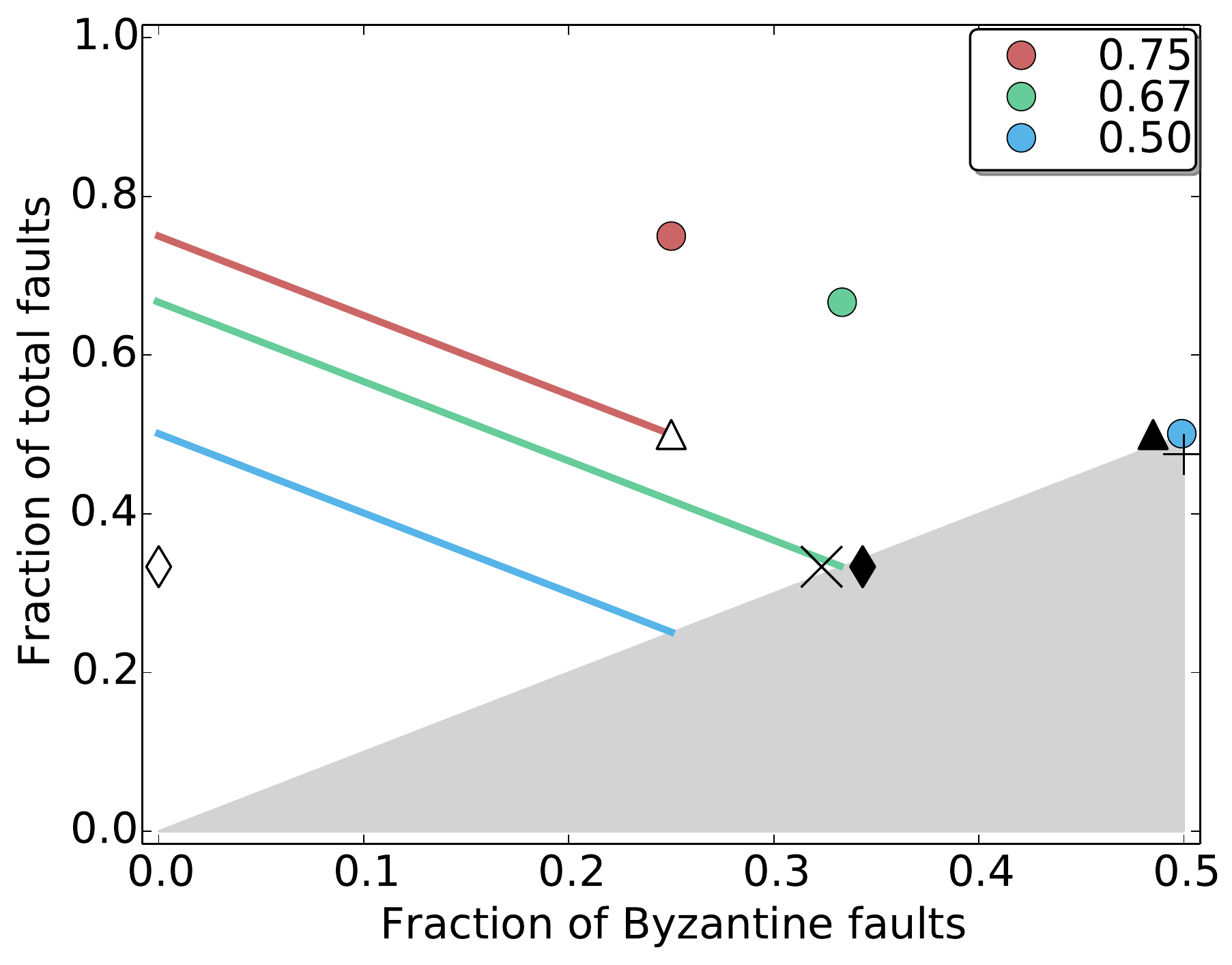}
  \end{subfigure}
  
  \begin{subfigure}{0.45\textwidth}
      \begin{tabular}{c l}
        $+$ & Partially Synchronous protocols~\cite{castro1999practical,Yin03,Martin06,Zyzzyva07,yin2018hotstuff,buchman2016tendermint}\\
        $\times$ & Synchronous Protocols~\cite{pass2018thunderella,hanke2018dfinity,abraham2019sync,abraham2018synchronous}\\
        $\blacktriangle$ & Thunderella, Sync HotStuff ($\vartriangle$: optimistic)~\cite{pass2018thunderella,abraham2019sync}\\
        $\blacklozenge$ & Zyzzyva, SBFT ($\lozenge$: optimistic)~\cite{Zyzzyva07}
      \end{tabular}
  \end{subfigure}
  \caption{\textbf{Comparing \name to existing consensus
      protocols.} The legend represent different $\qmin$ values.}
  \label{fig:compare}
\end{figure}

Most BFT protocols are designed with a uniform assumption about the
system. The literature on BFT consensus is vast and is largely beyond scope for
review here; we refer the reader to the standard
textbooks in distributed computing~\cite{lynch1996distributed,attiya2004distributed}.

\paragraph{Resilience.} Figure~\ref{fig:compare} compares resilience in \name
with some existing consensus protocols.  The x axis represents a Byzantine
resilience threshold, the y axis the total resilience against
corruption under the \psing fault mode.  The three
different colors (red, green, blue) represent three possible instantiations of
\name at different $\qmin$'s. 

Each point in the figure represents an abstract ``client'' belief. For the partial
synchrony model, client beliefs form lines, and for synchronous settings,
clients beliefs are individual circles.
The locus of points on a given color represents all client
assumptions supported for a corresponding $\qmin$, representing the diversity of
clients supported. 
The figure depicts state-of-art resilience combinations by existing consensus solutions via 
uncolored shapes, $+, \times, \vartriangle, \blacktriangle, \lozenge, \blacklozenge$.
Partially synchronous protocols~\cite{castro1999practical, yin2018hotstuff,buchman2016tendermint} that tolerate one-third Byzantine faults can all be represented  by the `+' symbol at $(1/3, 1/3)$. 
Similarly, synchronous
protocols~\cite{hanke2018dfinity,abraham2018dfinity,abraham2018synchronous}
that tolerate one-half Byzantine faults are represented by the
`$\times$' symbol at $(1/2, 1/2)$.
It is worth noting that some of these works 
employ two commit rules that differ in number of votes 
or synchrony~\cite{Martin06,Zyzzyva07,pass2018thunderella,abraham2019sync}.
For instance, Thunderella and Sync HotStuff optimistically commit 
in an asynchronous fashion based on quorums of size $\geq 3/4$, 
as represented by a hollow triangle at $(1/4, 1/2)$.
Similarly, FaB~\cite{Martin06}, Zyzzyva~\cite{Zyzzyva07} and SBFT~\cite{gueta2018sbft}
optimistically commit when they receive all votes but wait for two rounds of votes otherwise.
These are represented by two points in the figure. 
Despite the two commit rules, 
these protocols do not have client diversity, 
all parties involved (replicas and clients) make the same assumptions
and reach the same commit decisions.

\paragraph{Diverse client beliefs.}
A simple notion of client diversity exists in Bitcoin's probabilistic commit rule.
One client may consider a transaction committed after six confirmations while another may require only one confirmation.
Generally, the notion of client diversity has been discussed informally at public blockchain forums. 

Another example of diversity is considered in the XFT protocol~\cite{XFT}.
The protocol supports two types of clients: clients that assume crash faults under partial synchrony,
or clients that assume Byzantine faults but believe in synchrony.
Yet another notion of diversity is considered 
by the federated Byzantine consensus model and the Stellar
protocol~\cite{mazieres2015stellar}.
The Stellar protocol allows nodes to pick their own quorums.
Our \name approach instead considers diverse clients 
in terms of \psing adversaries and synchrony.
The model and techniques in~\cite{mazieres2015stellar} and our paper
are largely orthogonal and complementary.

\paragraph{Flexible Paxos.} Flexible Paxos by Howard et
al.~\cite{DBLP:conf/opodis/HowardMS16} observes that 
Paxos may use non-intersecting quorums within a view 
but an intersection is required across views.
Our Flexible Quorum Intersection (b) can be viewed as
its counterpart in the Byzantine and \psing setting.
In addition, \name applies the flexible quorum idea
to support diverse clients with different fault model and timing assumptions.

\paragraph{Mixed fault model.}
Fault models that mix Byzantine and crash faults have been considered in various
works, e.g., FaB~\cite{Martin06} and SBFT~\cite{abraham2019sync}. The \psing
faults are in a sense the opposite of crash faults, mixing Byzantine with
``anti-crashes''. 
Our \psing adversary bears similarity to a rational adversary in the BAR
model~\cite{aiyer2005bar}, with several important differences.
BAR assumes no collusion exists among rational replicas themselves and between
rational and Byzantine replicas, whereas \psing replicas have no such
constraint.
BAR solutions are designed to expose cheating behavior and thus deter rational
replicas from cheating. The \name approach does not rely on deterrence for good
behavior, and breaks beyond the $1/3$ ($1/2$) corruption tolerance threshold in
asynchronous (synchronous) systems. 
Last, BAR solutions address only the partial synchrony settings. 
At the same time, BAR provides a game theoretic proof of rationality. 
More generally, game theoretical modeling and analysis with collusion have been performed to other problems 
such as secret sharing and multiparty computation~\cite{abraham2006distributed,lysyanskaya2006rationality,gordon2006rational,kol2008cryptography}. 
Analyzing incentives for the \psing model remains an open challenge.


%% file: conclusion.tex
\section{Conclusion and Future Work}
\label{sec:conclusion}

We present Flexible BFT, a protocol that supports diverse clients with different assumptions to use the same ledger. Flexible BFT allows the clients to tolerate combined (Byzantine plus \pleasing) faults exceeding 1/2 and 1/3 for synchrony and partial synchrony respectively. At a technical level, under synchrony, we show a synchronous protocol where the replicas execute a network speed protocol and only the commit rule uses the synchrony assumption. For partial synchrony, we introduce the notion of Flexible Byzantine Quorums by deconstructing existing BFT protocols to understand the role played by the different quorums. We combine the two to form Flexible BFT which obtains the best of both worlds.

Our liveness proof in Section~\ref{sec:proof}
employs a strong assumption that all clients have correct commit rules.
This is because our \pleasing fault model did not specify 
what these replicas would do if they can violate safety for some clients.
In particular, they may stop helping liveness. 
However, we believe this will not be a concern 
once we move to a more realistic rational model.
In that case, the best strategy for \pleasing replicas
is to attack the safety of clients with unsafe commit rules 
while preserving liveness for clients with correct commit rules. 
Such an analysis in the rational fault model
remains interesting future work.
Our protocol also assumes that all replicas have clocks that advance at the same rate. It is interesting to explore whether our protocol can be modified to work with clock drifts.
